\documentclass[12pt]{article}
\usepackage[utf8]{inputenc}

\usepackage{amsmath,amsfonts,amssymb,latexsym,amsthm,verbatim,a4wide,url,appendix,pgf} 
\title{The COVID-19 pandemic as experienced by the individual}
\author{Patrick Garland\thanks{School of Medicine, University of Southampton, UK: email {\tt P.Garland@soton.ac.uk}}  
\and
Dave Babbitt\thanks{Data Scientist, Booz Allen Hamilton, USA: email {\tt babbitt\_dave@bah.com}} 
\and
Maksym Bondarenko\thanks{WorldPop, School of Geography and Environmental Science, University of Southampton, UK: email
{\tt M.Bondarenko@soton.ac.uk}} 
\and
Alessandro Sorichetta\thanks{WorldPop, School of Geography and Environmental Science, University of Southampton, UK: email
{\tt A.Sorichetta@soton.ac.uk}} 
\and
Andrew J. Tatem\thanks{WorldPop, School of Geography and Environmental Science, University of Southampton, UK: email
{\tt A.J.Tatem@soton.ac.uk}} 
\and
Oliver Johnson\thanks{School of Mathematics, University of Bristol, UK: email {\tt maotj@bristol.ac.uk}}
}
\date{\today}

\usepackage{graphicx}
\newcommand{\psk}{\mbox{ people/km}^2}

\newtheorem{theorem}{Theorem}[section]
\newtheorem{lemma}[theorem]{Lemma}

\newtheorem{definition}[theorem]{Definition}
\newtheorem{remark}[theorem]{Remark}
\newtheorem{example}[theorem]{Example}

\begin{document}

\maketitle

\begin{abstract}
The ongoing COVID-19 pandemic has progressed with varying degrees of intensity in individual countries, suggesting it is important to analyse  factors that vary between them. We study measures of `population-weighted density', which capture density as perceived by a randomly chosen individual. These measures of population density can significantly explain variation in the initial rate of spread of COVID-19 between countries within Europe. However, such measures do not explain differences on a global scale, particularly when considering countries in East Asia, or looking later into the epidemics.
Therefore, to control for country-level differences in response to COVID-19 we consider the cross-cultural measure of individualism proposed by Hofstede. This score can significantly explain variation in the size of epidemics across Europe, North America, and East Asia. Using both our measure of population-weighted density and the Hofstede score we can significantly explain half the variation in the current size of epidemics across Europe and North America.
By controlling for country-level responses to the virus and population density, our analysis of the global incidence of COVID-19 can help focus attention on epidemic control measures that are effective for individual countries.
\end{abstract}

\section{Introduction}

Since the onset of the COVID-19 pandemic in December 2019, over a million people have died of the disease worldwide 
\cite{who2020}. Given the vast human and economic cost, considerable effort has been given to understanding and modelling epidemics by the media and scientific community \cite{pressgazette,ferguson2020report}. At this early stage, the media in particular have performed the role of disseminating mortality data between and within countries, often through data visualizations. 

Such visualizations \cite{pressgazette} are informally used to assess the responses and interventions of various Governments, or the quality of their health services, and have therefore put pressure on governments to allocate resources based on apparent differences between the performance of countries. We caution against making such judgements based on raw data alone without controlling for factors that may impact the rate of spread of COVID-19 and the severity of any outbreak. These factors may include: demographic variables such as age and ethnicity, pre-existing medical conditions, patterns of mobility (i.e., travel), and cultural factors such as voluntary mask wearing, household age profile, or patterns of social contact. These factors exist independently of the quality of any specific Government coronavirus response, and more formal analysis is beginning to emerge that suggests they are important for the apparent severity of epidemics between countries \cite{chaudhry2020country}.

Any assessment of a COVID-19 epidemic in any country must be carried out carefully, in the context of these factors. It will be a major statistical exercise to do this properly, considering interactions rigorously. In this article, we make a minor early contribution to such an audit by considering one demographic factor, population density, while controlling for cross-cultural differences that may impact country-level responses to the virus.

Population density plays a significant role in emerging infectious diseases \cite{jones2008global, weiss, morse}. The basic reproduction number ($R_0$) for a zoonotic pathogen not yet fully adapted for human-to-human transmission may in fact require sufficiently many interactions, facilitated by population density, for an outbreak to occur \cite{karesh}. This need for a sufficient density of people susceptible to infection was apparent to those first modelling infectious disease outbreaks nearly 100 years ago (\cite{SIRpaper}). Contemporary modelling of influenza pandemics have utilised this framework, and subsequently employed it to model COVID-19 epidemics (\cite{ferguson2006strategies,ferguson2020report}). Of note in these models is the use of the standard measure of population density (total number of people divided by area) applied to relatively high resolution population data (square kilometre). This prompted us to consider two questions: does weighting high resolution population data by the local density experienced by an average individual improve its ability to predict epidemic growth, and could such a methodology be used to explain the dynamics of epidemics between countries? 

Regardless of what demographic variables are chosen to compare COVID-19 epidemics between countries, the same problem will be encountered: individual countries and their governments have responded differently to the virus. Public health interventions have included: testing, contact tracing, and isolation of positive cases and their contacts; restricting public gatherings, physical distancing, mandated face mask wearing, and closures of schools and nonessential businesses \cite{hartley}. It is noticeable that such measures appear to have been more effective in East Asian countries than Europe and North America; for example, China has reduced cases by 90\% but Italy and the USA have not replicated this performance \cite{remuzzi,parodi}. This may be partly due to the speed with which control measures have been instigated; and in particularly, when the strongest measure, lockdowns, have been ordered \cite{anderson,flaxman2020estimating}. Any analysis that attempts to compare country-level COVID-19 epidemics will need to control for these differences. The Oxford COVID-19 Government Response Tracker collates these country-level responses into a ‘Stringency index’ \cite{hale}; however, whether these government responses are derived from a more fundamental difference in cultural responses to infectious diseases has not been currently explored.

It has been hypothesised that pathogen exposure puts selection pressure on animals to modify their behaviour, which in humans has been labelled the `pathogen-stress hypothesis' \cite{fincher2008pathogen,fincher2019encyclopedia,morand2018}. Fincher et al. \cite{fincher2008pathogen} hypothesised that historical pathogen prevalence could explain variation in cross-cultural measures of individualism and collectivism. In particular, historical and contemporary pathogen prevalence was significantly negatively correlated with, for example, the cross-cultural measure of individualism proposed by Hofstede \cite{hofstede2005}, which captures to what extent a country values individual freedom over collective benefit. This axis between individualism and collectivism has been shown to be particularly stark between Western and East Asian countries, with the former more individualistic and the latter more collectivist. Interestingly, countries with high individualism scores also have a greater susceptibility to zoonotic disease outbreaks \cite{morand2018}. Whether these observations can be replicated for the current COVID-19 pandemic and used to control for country-level responses is currently an unanswered question.

We have used publicly available data to examine whether population density can explain the burden of COVID-19 disease between countries. We first explore two methods for rationalising the ‘lived density’ experienced by an average individual. High resolution population data from the WorldPop data set has  allowed us to calculate two measures of population-weighted density (PWD) for multiple countries across East Asia, Europe, North America, and US states. We investigate whether these PWD measures can explain the initial and current stages of the ongoing COVID-19 pandemic. Finally, we examine whether the Hofstede cross-cultural measure of individualism can account for variance between the current size of COVID-19 epidemics, and subsequently use both PWD and the Hofstede score to explain the current size of COVID-19 epidemics across a range of countries. This analysis provides a rationale for controlling two important variables, population density and country-level responses, when comparing epidemics between countries.

\section{Methods} \label{sec:methods}

\subsection{COVID-19 mortality data}

We use two sets of data in this paper to measure the extent of COVID-19 in a region, firstly a measure that captures the rate of spread early on, and secondly a measure that captures the size of the epidemic after a period of time has passed.

The first measure examines the rate of spread early in the epidemic by choosing a relatively small threshold value of deaths, and then comparing the number of deaths that have occurred a fixed number of days later. For concreteness, we choose these numbers to be 5 and 5: that is, our metric is ``number of deaths that have occurred 5 days after deaths hit 5''.  In the sense of the classical SIR model \cite{SIRpaper}, we anticipate that this is a phase of unrestricted exponential growth, where we seek to compare the growth exponent. Since this measure seeks to capture the early rate of spread, we do not believe it is appropriate to normalize these figures by total population size, since this corresponds to a situation of a limited number of small localized outbreaks which have not spread to the whole country.

Clearly the choice of these numbers is somewhat arbitrary: we want to choose a threshold that is large enough that daily numbers are not affected too much by random fluctuations, and to wait a long enough period for random daily effects to cancel out. However, waiting too long means that the numbers can be affected by Government actions such as lockdowns. Further, waiting for deaths to hit too high a level may exclude some smaller countries. Some experimentation has shown that our results are  robust to the choice of parameters. 

An alternative method would be to study the rate of growth of cases. However, such comparison risk being distorted by the significant variation between availability of tests in different countries early in the pandemic.

Our second measure is to consider total deaths per million 120 days after the 1st clinically confirmed case. We are therefore measuring epidemics at a comparable stage of their duration. 

The data on COVID-19 deaths itself was obtained online. Data for European countries was taken from \cite{ourworld}, which is itself based on ECDC data. Data for US states was downloaded from \cite{nytgit}. Data for East Asian countries was also taken from \cite{ourworld}.

\subsection{Population data}


We studied 30 European countries, using figures of standard population density $\rho_S$ (see Results for notation) taken from Wikipedia \cite{wikidensity}, and using the values of the non-empty lived density $\rho_N$ for each country calculated by Rae from Eurostat data, and stated in \cite{rae}. Some extremely small countries (where the size of the COVID-19 outbreak was too small to achieve the threshold discussed above) were excluded from the comparison.

Additionally, we used WorldPop population count data \cite{worldpop} to calculate the population-weighted density $\rho_W$ for these countries, using the formula \eqref{eq:rhoQ}. As described in \cite{sorichetta,stevens} WorldPop distribution datasets are produced by disaggregating administrative-unit based official population estimates into grid cells having a resolution of 3 arc seconds (approximately 100m at the equator) before summarizing them at 30 arc seconds (approximately 1km at the equator). To ensure comparability between countries, 2019 gridded country-based population estimates were adjusted to match the corresponding UNPD country total estimates \cite{untotals}.


We tabulate the resulting values $\rho_N$, $\rho_W$, $\rho_S$ and $V = \rho_W/\rho_S$ in Appendix \ref{sec:eurodata} for completeness. Additionally, in Appendix \ref{sec:miscdata} we consider a range of other countries of interest and  provide values of their standard population density $\rho_S$ (again taken from Wikipedia \cite{wikidensity}) and population-weighted density $\rho_W$ calculated from WorldPop data \cite{worldpop} in the same way.

We studied the 50 US States, using figures of standard population density $\rho_S$ taken from Wikipedia, and calculating the population-weighted density $\rho_W$ ourselves, again based on WorldPop data. Values of $\rho_N$ were not available for US states, but since we found a stronger relationship between rate of spread in Europe with $\rho_W$ than with $\rho_N$ (see Figures \ref{fig:EuroPlots2} and \ref{fig:EuroPlots3}), in any case we prefer to focus on $\rho_W$. We tabulate the density values for US states in Appendix \ref{sec:usdata} for completeness.



\subsection{Hofstede individualism score}

We used the values of this measure for a range of countries calculated from questionnaires and tabulated in \cite[Table 4.1]{hofstede2005}.

\subsection{Nomenclature of regions analysed, data processing and availability}

Global regions analysed were: Western/West (Europe, USA, Canada, Israel); East Asia and Oceania (Australia and New Zealand). Countries are listed in appendices \ref{sec:eurodata} and \ref{sec:miscdata}

Fits for normal or lognormal (base 10) distributions were generated using a maximum likelihood approach, and the likelihoods compared for statistically significant differences (see \cite[Section 6.7.2]{burnham2002model}). It was observed that the initial rate of spread of COVID-19 and later deaths per million associated with it followed a lognormal distribution, and this was used for regression analysis. The antilog is presented on graphs to aid interpretation.
All data is available at: {\tt https://github.com/ptrckgrlnd/COVID-PWD-and-Individualism}

\section{Results} \label{sec:plots}

\subsection{Lived density, or population density as experienced by a random individual}

First we motivate and define the various measures of population density that we use in this paper.

\begin{example} \label{ex:countries} Consider three countries, each with 100,000 people and an area of 100$km^2$, and consider the population of each square kilometre grid square. As illustrated in 
Figure \ref{fig:countries}:
\begin{enumerate}
    \item Averagia has a uniform spread of population, with 1,000 people living in each grid square. 
    \item Builtupia has ten towns, each consisting of 10,000 people living in a single square kilometre, and the remaining land is uninhabited.
    \item Citia has one city, where 100,000 people live in a single square kilometre, and the remaining land is uninhabited.
\end{enumerate}
\end{example}

\begin{figure}[h!]
\includegraphics[width=\textwidth]{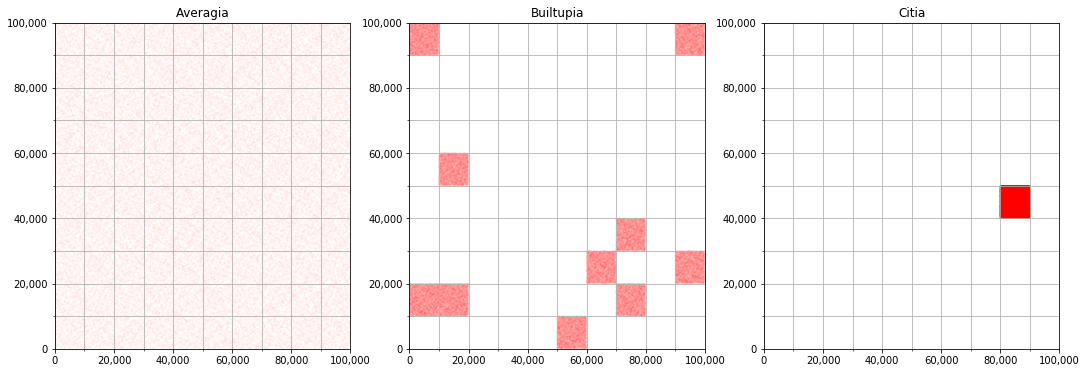}
\caption{Schematic diagram of population in Averagia, Builtupia and Citia -- see Example \ref{ex:countries}. \label{fig:countries}}
\end{figure}

Since each standard population density is $100000/100 =  1000 \psk$, we might naively say that each country is equally crowded. However, if we consider the lived experience of a person in each country, it is clear that in everyday life residents of Citia would have daily encounters with more people than in Builtupia, who in turn would have daily encounters with more people than Averagia. Since daily encounters are likely to cause infection events, it seems natural to imagine that COVID-19 would spread more rapidly in Citia than in Builtupia, and in turn faster in Builtupia than in Averagia. We would like to use a measure of population density that captures this, and here define two such measures.

We will consider a region of total area $A$, divided into $M$ subareas of area $A_1$, \ldots, $A_M$ respectively. We will write $n_i$ for the population of the $i$th subarea, and write $n = \sum_{i=1}^M n_i$ for the total population of the region.

\begin{definition} \label{def:spd}
In this notation, the {\em standard population density} is 
\begin{equation} \label{eq:rhoS} \rho_S = \frac{n}{A} = \frac{ \sum_{i=1}^M n_i}{\sum_{i=1}^M A_i}.\end{equation}
\end{definition}

\begin{definition} 
We will refer to the {\em non-empty lived density} \cite{rae} as the same expression normalized by the total area  where some people live. That is, we write $S$ for the set of regions with non-zero population, and consider
\begin{equation} \label{eq:rhoN}
 \rho_{N} = \frac{ \sum_{i \in S} n_i}{\sum_{i \in S} A_i}.
 \end{equation}
\end{definition}

Note that in comparison with the expression \eqref{eq:rhoS} given for $\rho_S$, the numerator is unchanged here, since removing terms which are zero does not affect the value of a sum. Indeed, if each subarea contains at least one person then the non-empty lived density $\rho_N$ coincides with the standard population density $\rho_S$.

\begin{definition} \label{def:qld}
We will refer to the {\em population-weighted density} \cite{craig1984averaging} as the sum
\begin{equation} \label{eq:rhoQ}
\rho_W = \frac{1}{n} \sum_{i=1}^M \frac{n_i^2}{A_i} 
= \sum_{i=1}^M \left( \frac{n_i}{n} \right) \left( \frac{n_i}{A_i} \right).
\end{equation}
\end{definition}

We can think of this as a quadratic measure because if each subarea has $A_i=1$ then it becomes $\rho_W = \left( \sum_{i=1}^M n_i^2 \right)/n$. Note that the form of this quantity is somewhat reminiscent of the Herfindahl--Hirschman Index, which is used in economics to measure diversity of market share (see \cite{hirschman1964american} for a history of this quantity).

\begin{remark} \label{rem:area} As described for example in \cite{ottensmann2018population}:
\begin{enumerate}
    \item We can think of this population-weighted density $\rho_W$ as follows: select an individual uniformly at random from the population (with probability $n_i/n$ this will be someone from subarea $i$). Then $n_i/A_i$ is the density of their local subarea. Hence, assuming the subareas are small enough to be reasonably homogeneous, this measure represents the expected local density, sampling by person. 

\item In contrast,  the standard population density corresponds to sampling by area. That is, pick a point uniformly in the region (with probability $A_i/A$ this will be a point in subarea $i$). Then according to this distribution, the expected local density will be 
\begin{equation} 
\sum_{i=1}^M \left( \frac{A_i}{A} \right) \left( \frac{n_i}{A_i} \right) = \frac{n}{A},\end{equation}
the standard population density $\rho_S$.
\end{enumerate}
\end{remark}

We return to Example \ref{ex:countries}, where recall that each country had standard population density $\rho_S = 1000$. Considering the square kilometre grid squares as the subareas, note that that
the non-empty lived density $\rho_N$ of Averagia, Builtuipia and Citia is $100000/100 = 1000 \psk$, $100000/10 = 10000 \psk$ and $100000/1 = 100000 \psk$ respectively. In fact, calculation shows that the same three values hold for the population-weighted  density $\rho_W$ in these cases (though this will not be true in general). For example, for Builtupia, there are 10 non-zero terms in the sum of
$\left( \frac{n_i}{n} \right) \left( \frac{n_i}{A_i} \right)$, each equal to $(1/10)( 10000/1)$, giving 10000 overall.

Notice that $\rho_S$, $\rho_N$ and $\rho_W$ are all measured in the same units, namely $\psk$, meaning that it is legitimate to compare them. It turns out that the three measures are always ordered in the same way for every region:

\begin{lemma} \label{lem:compare}
For any region, the three measures satisfy
$$ \rho_S \leq \rho_N \leq \rho_W,$$
with equality in the first inequality if and only if all the areas have non-zero population, and with equality in the second inequality if and only if all the areas of non-zero population have the same density.
\end{lemma}
\begin{proof}
 The fact that $\rho_S \leq \rho_N$ is clear, since the two expressions both have the same numerator $n$, but the non-empty lived density has a smaller denominator since the sum is taken over a smaller range. Clearly, equality holds if and only if the denominators are equal. 

The fact that $\rho_N \leq \rho_W$ is more subtle, but can be proved using the Cauchy-Schwarz inequality. Recall that we write $S$ for the set of areas of non-zero population, then use the fact that
\begin{equation}
 n^2 = \left( \sum_{i \in S} \frac{n_i}{\sqrt{A_i}} \sqrt{A_i} \right)^2 \leq \left( \sum_{i \in S} \frac{n_i^2}{A_i} \right) \left( \sum_{i \in S} A_i \right) = (n \rho_W) \left( \sum_{i \in S} A_i \right) ,\end{equation}
and rearrange, using the fact that the value of $\rho_W$ is unchanged if we restrict to a sum over $i \in S$ (since terms with
$n_i = 0$ do not contribute to the sum). Equality holds if and only if each populated subarea has the same density, that is if $n_i/A_i$ does not depend on $i$ for each $i \in S$.
\end{proof}

Indeed, a version of the same argument shows that if we partition a subarea into smaller subareas (by obtaining more finely grained population and area data) then the population-weighted density $\rho_W$ does not decrease, and will strictly increase unless all the new subareas have the same density. In other words, the maximum value of $\rho_W$ for a region is obtained by breaking it into small homogeneous subareas. 

\begin{definition} \label{def:regcoeff}
For a given region we define the {\em variability coefficient} $V = \rho_W/\rho_S$. Note that this is a dimension-free quantity.
\end{definition}

From Lemma \ref{lem:compare}, we know that  $V \geq 1$ for any region. This ratio measures the extent to which population is evenly distributed in a region; returning to the toy Example \ref{ex:countries}, notice that Averagia has $V=1$, Builtupia has $V=10$ and Citia has $V=100$. These values may help calibrate our understanding of the values of $V$ observed for actual regions.



\subsection{Population density analysis of Western countries}

First, since we have values for all three measures $\rho_S$, $\rho_N$ and $\rho_W$ for a range of European countries, we compare
how well each measure explains the rate of spread of COVID-19.

\begin{figure}[h!]
    \centering
    \includegraphics[width = 0.7\textwidth]{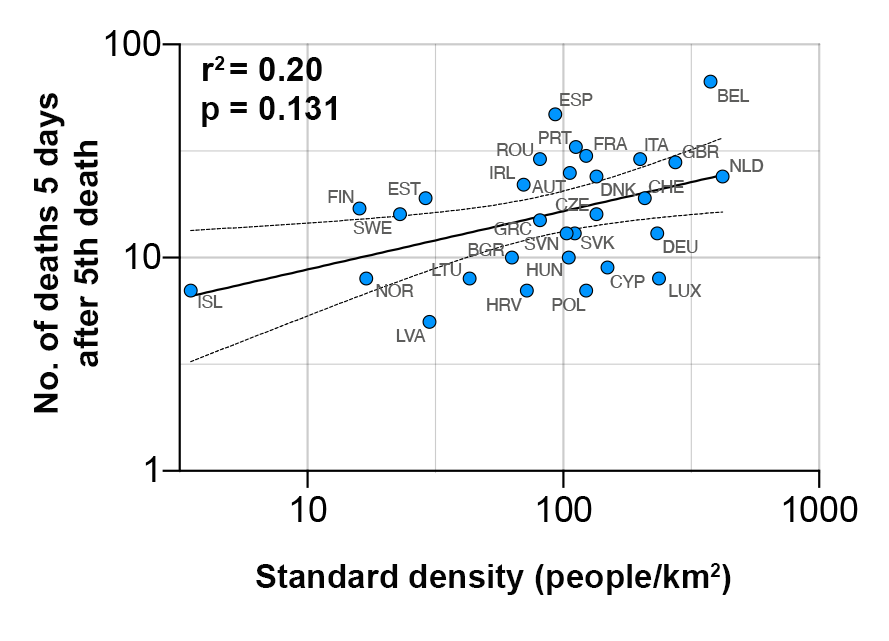}
     \caption{Simple linear regression analysis of rate of COVID-19 spread against standard  population density $\rho_S$ for European countries. Both axes are on a log scale. $n=30$, dashed lines represent 95\% CI. 
    \label{fig:EuroPlots1}}
\end{figure}

In Figure \ref{fig:EuroPlots1} we plot the correlation  between the rate of COVID-19 spread and the standard population density $\rho_S$ for European countries. However, the $r^2$ of $0.20$ is relatively low and the $p$-value of $0.131$ means a lack of  statistical significance. In contrast, for the rate of spread and the respective lived density measures, the effect is stronger than for the corresponding plot for $\rho_S$. In Figure \ref{fig:EuroPlots2} we see a significant correlation between rate of spread and $\rho_N$ for Europe ($r^2 = 0.26$ and $p < 0.004$), and in Figure \ref{fig:EuroPlots3} we see an even stronger effect in terms of $\rho_W$ ($r^2 = 0.53$ and $p < 0.0001$). Therefore, we conclude that in each case using non-standard measures of population density reveals statistically significant effects compared with those arising from the standard population density $\rho_S$. It appears that the population-weighted measure $\rho_W$ is of the most value in this sense, explaining over half the variation in the rate of spread of COVID-19 across Europe.

\begin{figure}[h!]
    \centering
    \includegraphics[width = 0.7\textwidth]{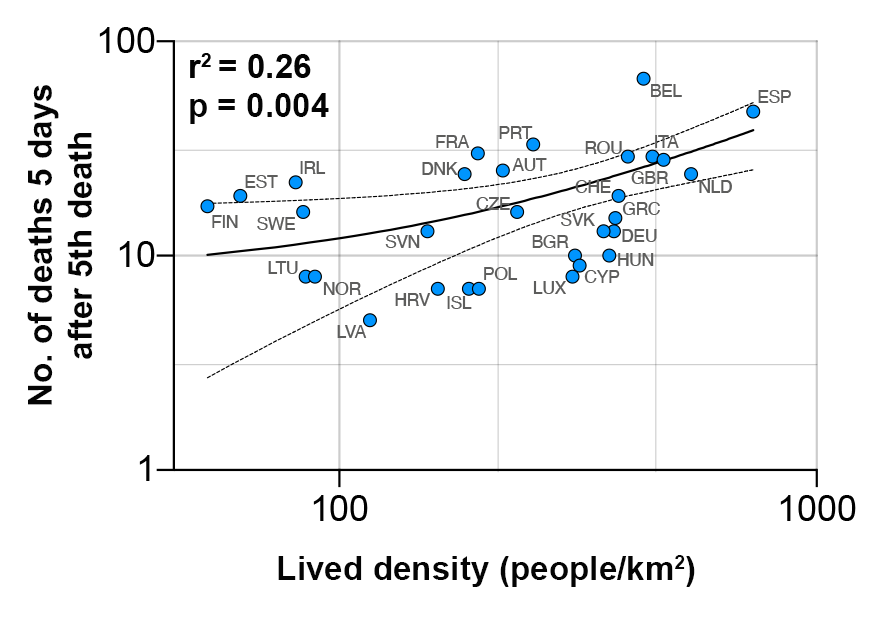}
     \caption{Simple linear regression analysis of rate of COVID-19 spread against non-empty lived  population density $\rho_N$ for European countries. Both axes are on a log scale. $n=30$, dashed lines represent 95\% CI.
    \label{fig:EuroPlots2}}
\end{figure}

\begin{figure}[h!]
    \centering
    \includegraphics[width = 0.7\textwidth]{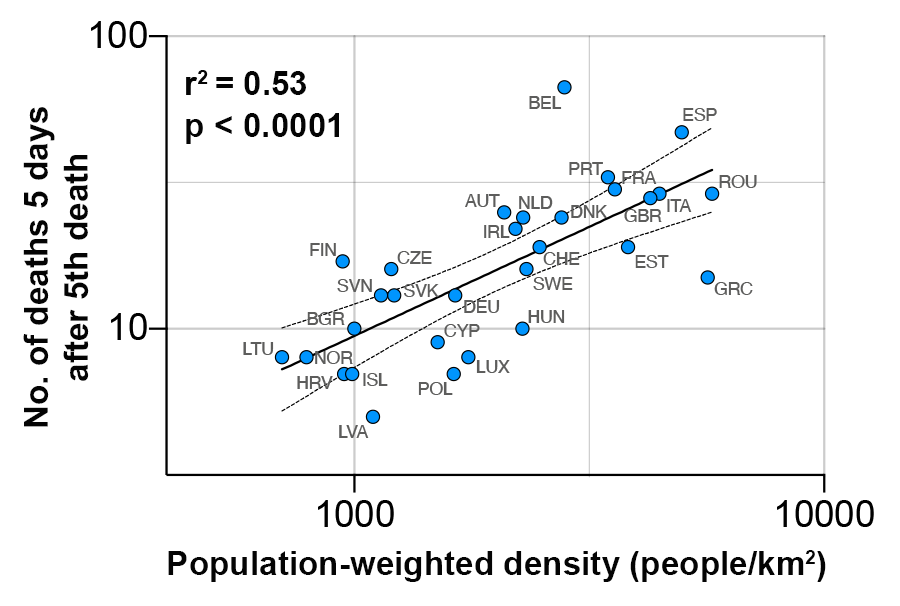}
     \caption{Simple linear regression analysis of rate of COVID-19 spread against population-weighted density $\rho_W$ for European countries. Both axes are on a log scale. $n=30$, dashed lines represent 95\% CI.
    \label{fig:EuroPlots3}}
\end{figure}

Having established that $\rho_W$ works well to explain the variation in rates of spread among European countries, we can perform a similar analysis for a wider range of countries.

First, we add the United States, Canada and Israel to the range of countries considered, to emphasise that the correlation observed is not restricted to European countries. As before, the standard density $\rho_S$ does not explain a significant fraction of the correlation, achieving $r^2 = 0.13, p = 0.0378$ (figure not shown here for brevity).

\begin{figure}[h!]
    \centering
    \includegraphics[width = 0.7\textwidth]{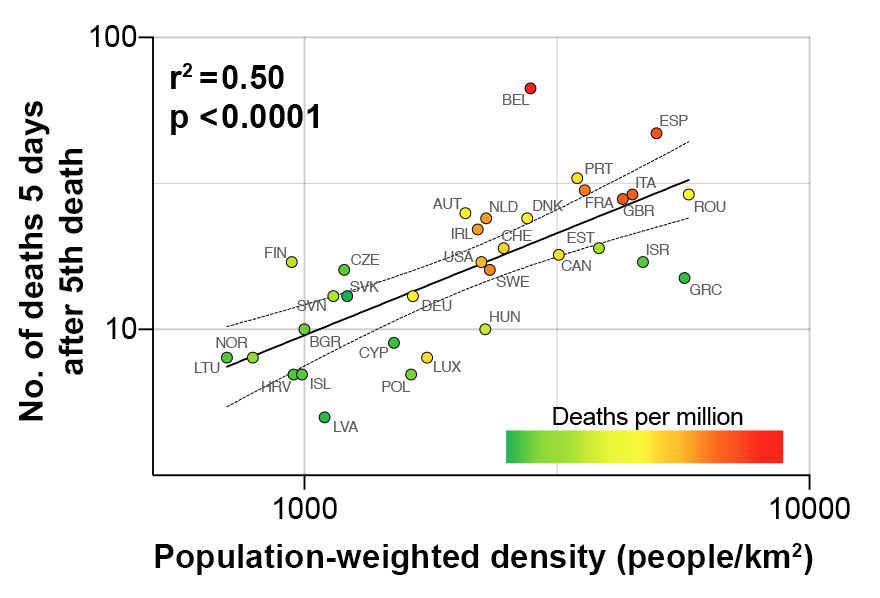}
    \caption{Simple linear regression analysis of rate of COVID-19 spread against population-weighted density $\rho_W$  for a range of Western countries. 
    Plot is on a log-log scale. $n=33$, dashed lines represent 95\% CI. Deaths per million are 120 days after the 1st confirmed case for each country.}
    \label{fig:WesternPlot}
\end{figure}

However, in Figure \ref{fig:WesternPlot}, we show that adding these further countries still leaves a good fit between $\rho_W$ and the number of deaths 5 days after the 5th death. Further, we show that 
this early rate of spread is crucial in determining the ultimate scale of the epidemic, by shading each point according to the number of deaths per million 120 days after the 1st case.  It is striking that generally the countries with the highest final death tolls are positioned to the right-hand side of the graph, indicating that they are typically those with large population-weighted densities.

\subsection{Limitations of population density: US states}

Next, we perform a similar analysis for the US states, to confirm analysis of \cite{riley} who found a statistically significant effect for population-weighted density $\rho_W$. 

\begin{figure}[h!]
    \centering
    \begin{tabular}{cc}
a) &     \includegraphics[width = 0.7\textwidth]{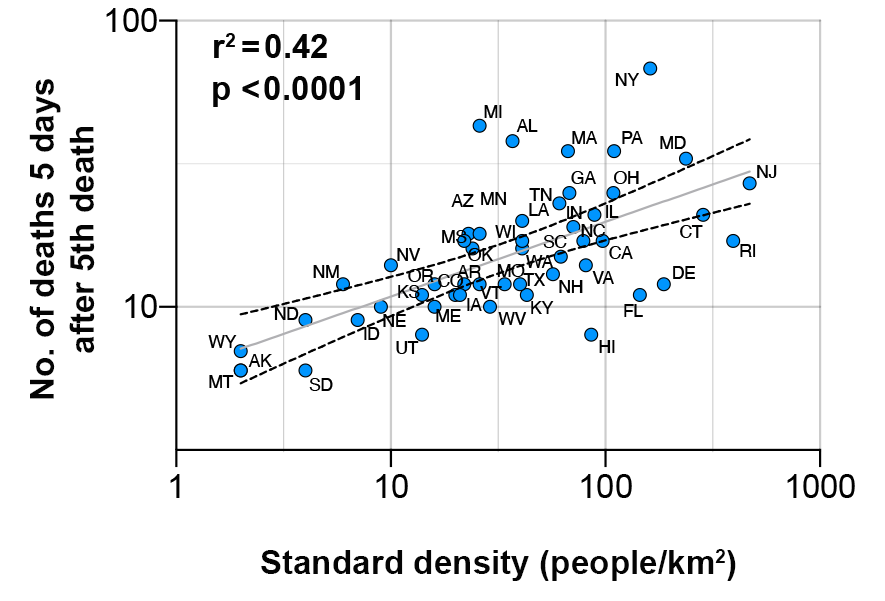} \\
b) &     \includegraphics[width = 0.7\textwidth]{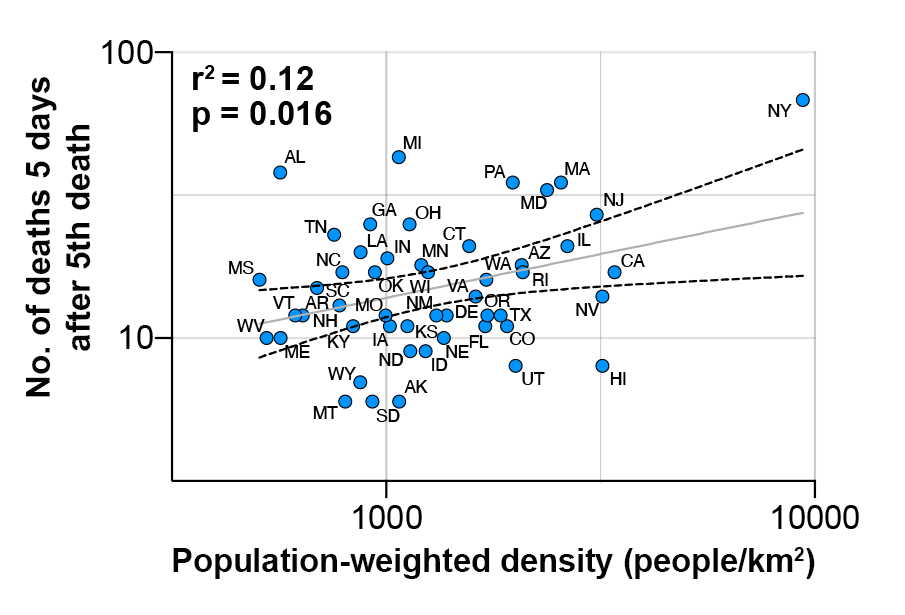}
 \end{tabular}
    \caption{Simple linear regression analysis for US states of rate of COVID-19 spread against
    a) standard density $\rho_S$ b)
    population-weighted density $\rho_W$. 
    Plots are on a log-log scale. $n=50$, dashed lines represent 95\% CI.}
    \label{fig:USPlot}
\end{figure}

As shown 
in Figure \ref{fig:USPlot}, the population-weighted density does indeed provide a statistically significant effect ($r^2 = 0.12$, $p = 0.016$). However it is striking that in this case, in contrast to the case of European countries, the standard density $\rho_S$ explains more of the variation ($r^2 = 0.42$, $p <0.0001$).

In particular, we can observe from Figure \ref{fig:USPlot} that there is surprisingly little variability in the values of $\rho_W$ calculated. For example, even extremely sparsely-populated states such as Alaska and  Wyoming (with standard density $\rho_S \leq 2$) have a population-weighted density $\rho_W$ between 800 and 1100, comparable with denser states such as Georgia ($\rho_S = 68$) and Indiana ($\rho_S = 71$).

\begin{figure}[h!]
    \centering
    \includegraphics[width = 0.7\textwidth]{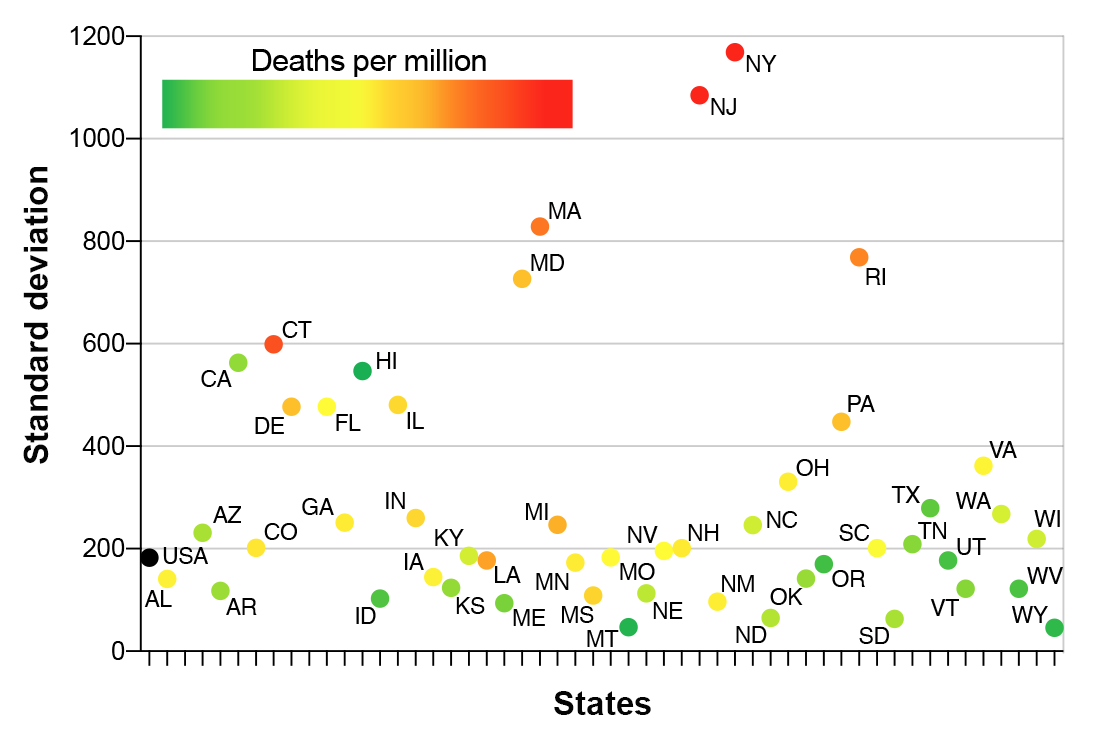}
    \caption{Values of standard deviation observed for US states, shaded according to the deaths per million 120 days after the 1st confirmed case in each state}
    \label{fig:USPlot2}
\end{figure}

We can understand this better by studying the standard deviation of the population values $n_i$ within grid squares in a state. If the squares all have area of one square kilometer, then calculation in the spirit of Remark \ref{rem:area} shows that since the expected value of $n_i$ is $\rho_S$, the variance of $n_i$ satisfies
$$ {\rm Var}(n_i) = \rho_S \rho_W - \rho_S^2,$$
so that the standard deviation is $\rho_S \sqrt{V-1}$, meaning it captures a combination of standard density and variability. 

In Figure \ref{fig:USPlot2} we plot the values of standard deviation for each state in alphabetical order, with points shaded according to their deaths per million 120 days after the first confirmed case in each state. It is striking to observe that high values of this standard deviation measure correctly identify the states with highest numbers of deaths - with the majority, however, having similar values for standard deviation and deaths per million.

\subsection{The Hofstede measure of individualism can control for country-level differences in the response to COVID-19}

Population-weighted density could not explain variation in the early or late stages of COVID-19 epidemics between multiple East Asian countries (data not shown). However, population-weighted density was able to significantly explain the number of deaths per million 120 days after the 1st confirmed case across a global collection of countries including both Western and East Asian countries but the variance explained was relatively low ($r^2 = 0.14$, $p = 0.0075$). This lack of explanatory power for population density (either $\rho_S$ or $\rho_W$) is particularly noticeable as multiple East Asian countries have very high population densities. For example, Singapore with a population-weighted density of 21,925 and Hong Kong with 46,301 are far in excess of any Western country. As both pre- and peri-pandemic responses to new and emerging infectious diseases could be affected by cultural differences we investigated the explanatory power of the Hofstede individualism score, which has been correlated with historical pathogen exposure \cite{fincher2008pathogen,fincher2019encyclopedia,morand2018}

 \begin{figure}[h!]
    \centering
    \includegraphics[width = 0.7\textwidth]{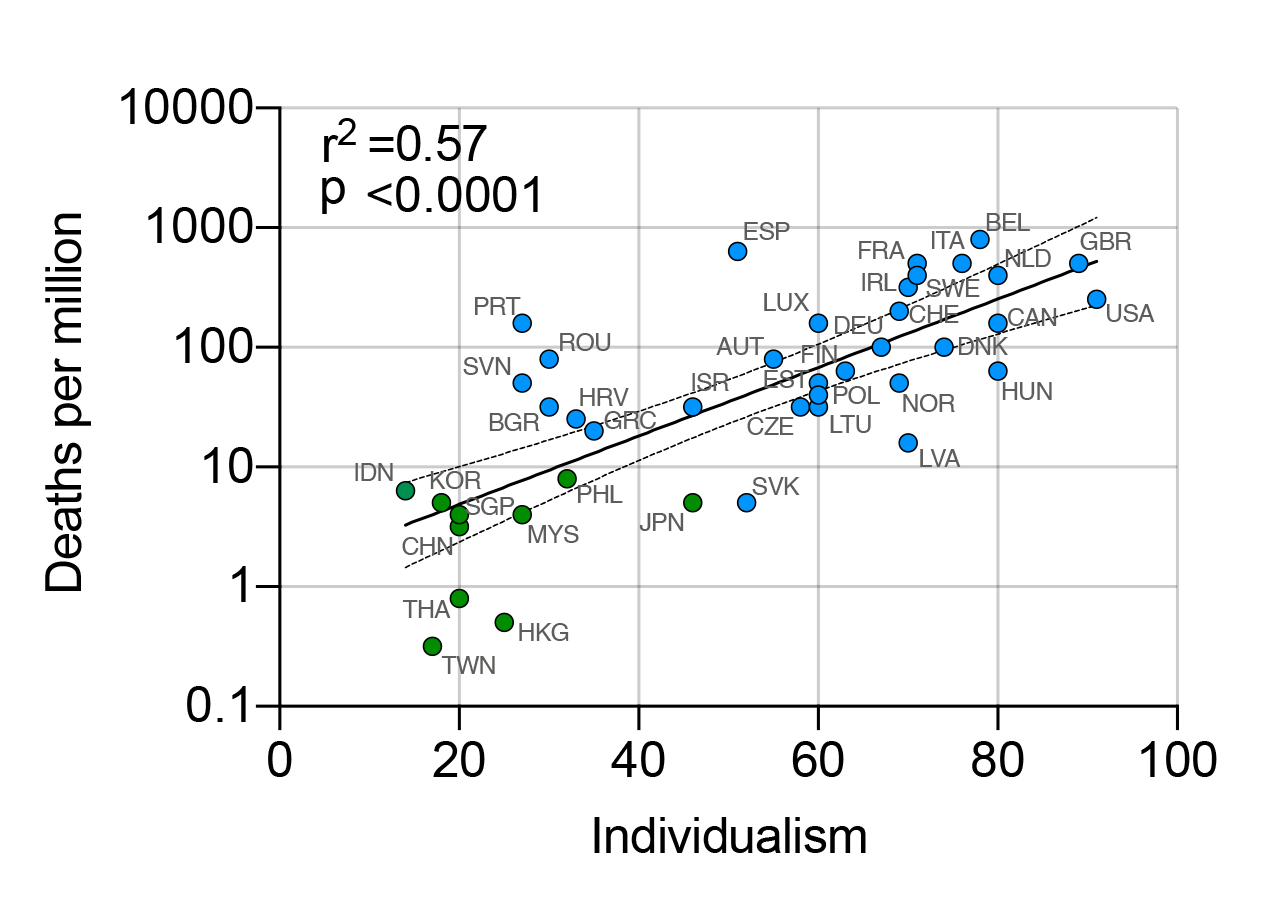}
    \caption{Simple linear regression analysis shows the Hofstede measure of individualism explains the global variation in deaths per million 120 days after the 1st confirmed COVID-19 case (green = Asian countries, blue = Western). $n=41$, dashed lines represent 95\% CI.   }
    \label{fig:individual}
\end{figure}

As can be seen in Figure \ref{fig:individual}, the Hofstede measure of individualism significantly explains variation in the number of deaths 120 days after the 1st clinically confirmed COVID-19 case across multiple Western and East Asian countries. 

One possible criticism of the individualism measure in this context is that the green and blue points are somewhat separated in Figure \ref{fig:individual}, and so this score may be serving as a proxy for other factors that give the difference between Western and East Asian countries. However, we observed that the Hofstede individualism score can control for country-level responses when examining the role of population density in explaining variation in COVID-19 epidemic sizes across Western countries late in the ongoing pandemic. For example, Figure \ref{fig:individual2} shows that individualism can improve the fit of deaths per million across Western countries. A multiple linear regression based on both population-weighted density $\rho_W$ and the Hofstede individualism score shows that both these variables are statistically significant in this context, and together give $R^2 = 0.46$ when considering deaths per million 120 days after the 1st confirmed COVID-19 case (which is considerably higher than the $r^2 = 0.27$ achieved in this context by $\rho_W$ alone).

 \begin{figure}[h!]
    \centering
    \includegraphics[width = 0.7\textwidth]{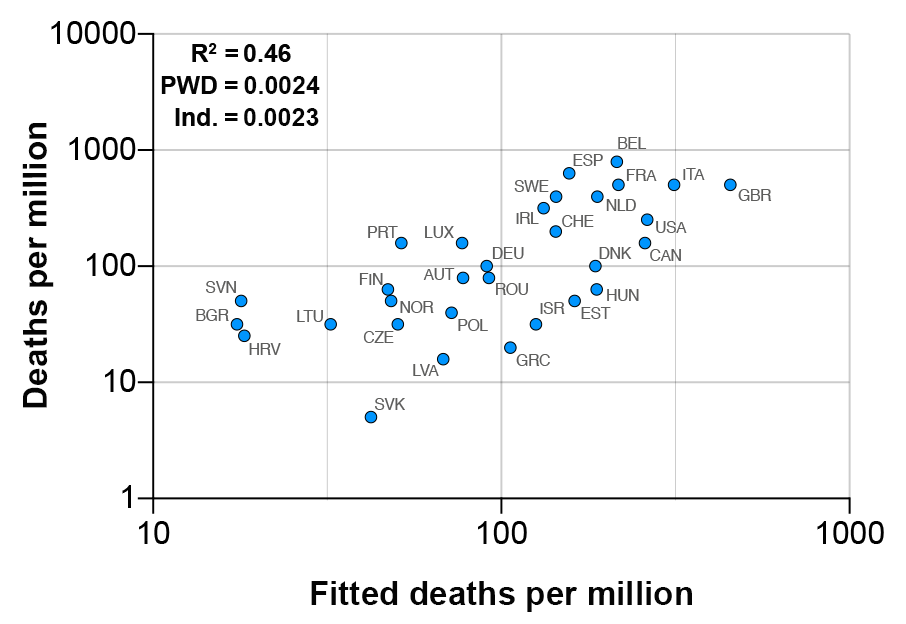}
    \caption{Multiple linear regression between deaths per million 120 days after the 1st confirmed COVID-19 case and a combination of the Hofstede individualism score (Ind) and population-weighted density $\rho_W$ (PWD) for Western countries. $n=33$.  }
    \label{fig:individual2}
\end{figure}

\section{Discussion} \label{sec:conclusion}

Country-level comparisons have been used to assess government responses to the COVID-19 pandemic. For these comparisons to be useful, they must be done using carefully controlled analysis. This study is an early contribution to this effort which focuses on population density. We find that the initial rate of spread of COVID-19 across Western countries can be readily explained by their population-weighted density. PWD does not appear to be as crucial for the spread of COVID-19 in East Asia, or for explaining the size of current epidemics recorded across multiple countries. However, controlling for country-level responses using the Hofstede cross-cultural measure of individualism allows a substantial proportion of the variance in the size of current COVID-19 epidemics across Western countries to be explained.

PWD can be seen to more readily explain the initial rate of COVID-19 spread by examining individual countries as plotted in Figure \ref{fig:EuroPlots3}. For example, Spain has one of the quickest rates of COVID-19 spread yet its standard population density $\rho_S = 93 \psk$ is relatively low; in contrast, Spain has one of the highest lived non-empty population densities in Europe, with a value of $\rho_N = 737 \psk$, and one of the highest PWD, with $\rho_W = 4980 \psk$. As discussed in detail in \cite{rae}, much of Spain is unpopulated, with very high density populations in Barcelona and Madrid. As a result, other than three very small countries (Monaco, Andorra and Malta, not studied here because of the size of their epidemics), Spain is one of the most densely populated countries in Europe, which our measure of population density captures. 

Also of note, are countries whose rate of spread is as predicted by their PWD despite popular perception about their relative performance. For example, Germany has a population-weighted density $\rho_W = 1641\psk$ –- lower than Sweden or Ireland (see Appendix \ref{sec:eurodata}) –- reflecting the fact that its population is fairly evenly distributed across the country (giving its relatively low value $V=7$). Considering Germany’s PWD, it may be unnecessary to postulate the existence of an `immunological dark matter'  to explain the current size of its epidemic \cite{friston}. In fact, from Figure \ref{fig:EuroPlots3} the European country which stands out for having a significantly slower spread than expected is Greece, which has a very high population-weighted density. This may be due to its high stringency national response \cite{hale}, which is particularly interesting as it has one of the lowest Hofstede individualism scores within Europe \cite[Table 4.1]{hofstede2005}.

Our observation that standard population density $\rho_S$ but not PWD $\rho_W$ predicts the rate of COVID-19 spread across US states was unexpected. We can observe from Figure \ref{fig:USPlot}b that there is surprisingly little variability in the values of $\rho_W$ calculated.  The fit observed in Figure \ref{fig:USPlot}a due to standard population density $\rho_S$ comes through sparsely-populated states such as Alaska and  Wyoming, with the remainder of the states forming a relatively homogenous cluster. It is striking that New York, the state with  the fastest rate of spread, has a significantly higher value of $\rho_W$ than any other state, whereas is not an outlier in terms of $\rho_S$, indicating the importance of $\rho_W$ in this context. These observations combined suggest that the lived population density experienced by many Americans is similar; this hypothesis is supported by the observation that the standard deviations for population density within US states as plotted in Figure \ref{fig:USPlot2} are also similar. It may be possible that other factors such as mobility (i.e., travel between states) affect the rate of COVID-19 spread within the USA; as has been observed, for example, in China \cite{lai}. Although, for epidemic management, there is still value in observing that US states with the highest PWD, like New York, are, as predicted, those with quickest rate of COVID-19 spread. For the USA as a whole, the Hofstede measure of individualism was found to accurately predict its current epidemic size relative to other less individualistic countries. Why might this cross-cultural axis between individualism and collectivism explain variation in COVID-19?

As argued by Fincher and Thornhill \cite{fincher2019encyclopedia}, collectivism has a two-fold benefit in protecting societies against pathogens. The first is greater attention to in-group vs out-group members, which can limit contact with out-group members either in general or during disease outbreaks; and the second is increased social pressure to conform to group priorities over individual rights. These are culture-wide attitudes and we hypothesise that the Hofstede  score can be used as a proxy for both government-level responses and everyday social pressure to perform prosocial behaviour. At the level of government, these differences in cultural values may have impacted how quickly and strongly control measures were put in place to limit the spread of COVID19 \cite{hale}. For example, limiting travel within and between countries for both nationals and non-nationals, widespread population surveillance to identify infected individuals and their contacts, and quarantining infected individuals and anyone they have come into contact with; in contrast, in more individualistic countries governments have been slow to restrict individual rights and less comprehensive \cite{hale,hartley}. At the level of everyday individual conduct the ability of a society to enforce prosocial behaviour (e.g., symptom vigilance and self isolation, mask wearing, hand hygiene) may impact the probability of outbreaks; for example, where they are dependent on superspreader events \cite{lloydsmith}.

We have focused on population density as a default, stable variable that may impact differences in the dynamics of epidemics between countries. What other variables might also play a role?

Advanced age is a risk factor for severe disease and mortality following infection with SARS-CoV-2 \cite{onder,wu}. Considering immunosenescence and ‘inflamm-aging’ are known features of advanced age, it is possible that industrialised countries with aging populations may have sufficient densities of immunocompromised people to promote outbreaks of zoonotic pathogens initially ill adapted to human-to-human transmission \cite{franceschi2000, franceschi2014, weiss}). And in terms of country comparisons, age has been a factor impacting relative differences in mortality \cite{chaudhry2020country}.

At the country-level and in the absence of physical distancing measures, mobility has a significant impact on the spread of COVID-19 (\cite{lai}). The degree to which this factor has impacted the international spread of COVID-19 has yet to be formally investigated. However, two countries noticeably absent from our analysis may give possible insight to this question. Although Australia and New Zealand only marginally reduce the variance explained for the initial rate of spread of COVID-19 amongst Western countries ($r^2 = 0.47$, $p<0.0001$), these two highly individualistic countries are outliers in the global comparison of deaths per million using the Hofstede score ($r^2 = 0.41$, $p<0.0001$). Compared to the UK, for example, AUZ and NZ have 4-fold and 10-fold less international arrivals, respectively, per year \cite{worldbank}. Therefore, it may be the case that the international connectivity of countries such as the UK has played an important role in the size of its epidemic \cite{pybus}.

\section {Conclusion}
This study attempts to robustly control one demographic variable, population density, so that relative differences in the spread of COVID-19 between countries can be compared. By highlighting where and when this factor has contributed to the spread of COVID-19, we believe this study can contribute to the discussion about which epidemic control measures are suitable for which countries. In general, shielding immunocompromised people living at relatively high density, for example the elderly in care facilities, should be essential. Other control measures exist over different time-scales. For example, in the short-term highly internationally connected countries such as the UK and USA should limit travel. However, in the long-term these two countries, as well as others across the West, should discuss democratically how to respond to new and emerging infectious diseases rapidly while preserving individual rights. As the prospect of a second wave of COVID-19 is likely \cite{xu}, both the short and long term, however, are now and soon.

\section*{Acknowledgments}

The authors thank Simon Johnson, Richard Pinch, Jenny Dewing, and Nathanael Hozé for useful discussions. The WorldPop project is funded through the Bill \& Melinda Gates Foundation (OPP1134076).

\newpage
\bibliographystyle{abbrv}


\newpage

\appendix

\section{Table of densities for European Countries} \label{sec:eurodata}

\begin{center}
\begin{tabular}{c|r|r|r|r}
Country & Non-empty $\rho_N$ &
Population-weighted $\rho_W$ & Standard $\rho_S$ & Variability $V= \rho_W/\rho_S$ \\
\hline 
Austria  &  220  &  2084  &  106  &  20 \\
Belgium  &  434  &  2804  &  376  &  7 \\
Bulgaria  &  312  &  1000  &  63  &  16 \\
Croatia  &  161  &  952  &  72  &  13 \\
Cyprus  &  319  &  1504  &  149  &  10 \\
Czechia  &  236  &  1198  &  135  &  9 \\
Denmark  &  183  &  2761  &  135  &  20 \\
Estonia  &  62  &  3827  &  29  &  132 \\
Finland  &  53  &  944  &  16  &  59 \\
France  &  195  &  3590  &  123  &  29 \\
Germany  &  376  &  1641  &  233  &  7 \\
Greece  &  379  &  5654  &  81  &  70 \\
Hungary  &  368  &  2280  &  105  &  22 \\
Iceland  &  187  &  989  &  3.5  &  283 \\
Ireland  &  81  &  2203  &  70  &  31 \\
Italy  &  453  &  4460  &  200  &  22 \\
Latvia  &  116  &  1096  &  30  &  37 \\
Lithuania  &  85  &  702  &  43  &  16 \\ 
Luxembourg  &  308  &  1750  &  237  &  7 \\
Netherlands  &  546  &  2289  &  420  &  5 \\
Norway  &  89  &  791  &  17  &  47 \\
Poland  &  196  &  1627  &  123  &  13 \\ 
Portugal  &  255  &  3469  &  112  &  31 \\
Romania  &  402  &  5771  &  81  &  71 \\
Slovakia  &  358  &  1215  &  111  &  11 \\
Slovenia  &  153  &  1140  &  103  &  11 \\
Spain  &  737  &  4980  &  93  &  54 \\
Sweden  &  84  &  2246  &  23  &  98 \\
Switzerland & 385 & 2479 & 208 & 12 \\
UK & 478 & 4265 & 274 & 16 \\
\end{tabular}
\end{center}

\section{Table of densities for other selected countries} \label{sec:miscdata}

\begin{center}
\begin{tabular}{c|r|r|r}
Country  &
Population-weighted  $\rho_W$ & Standard $\rho_S$ & Variability $V = \rho_W/\rho_S$ \\
\hline 
Australia    &    2273  &  3.2  &  758 \\
Canada    &    3188  &  4  &  797 \\
China    &    7098  &  146  &  49 \\
Hong Kong    &    46301  &  6781  &  7 \\
Israel    &    4678  &  417  &  11 \\
Japan    &    5338  &  333  &  16 \\
Malaysia    &    2745  &  99  &  28 \\
New Zealand    &    2030  &  19  &  107 \\
Philippines    &    9254  &  362  & 26 \\
Singapore    &    21925  &  7894  &  3 \\
South Korea    &    9265  &  516  &  18 \\
Indonesia       &   3447  &  141  &  24 \\
Taiwan    &    10602  &  652  &  16 \\
Thailand    &    2832  &  130  &  22 \\
USA    &    2241  &  34  &  66 \\
\end{tabular}
\end{center}

\section{Table of densities for US States} \label{sec:usdata}

\begin{center}
\begin{tabular}{c|r|r|r}
State &	Population-weighted $\rho_W$ & Standard $\rho_S$ & Variability $V= \rho_W/\rho_S$ \\
\hline
AL  &  566  &  37  &  15.3 \\
AK  &  1072  &  0.5  &  2144.6 \\
AZ  &  2070  &  23  &  90.0 \\
AR  &  639  &  22  &  29.0 \\
CA  &  3409  &  97  &  35.1 \\
CO  &  1913  &  20  &  95.7 \\
CT  &  1562  &  286  &  5.5 \\
DE  &  1385  &  187  &  7.4 \\
FL  &  1705  &  145  &  11.8 \\
GA  &  919  &  68  &  13.5 \\
HI  &  3198  &  86  &  37.2 \\
ID  &  1236  &  7  &  176.6 \\
IL  &  2650  &  89  &  29.8 \\
IN  &  1007  &  71  &  14.2 \\
IA  &  1020  &  21  &  48.6 \\
KS  &  1120  &  14  &  80.0 \\
KY  &  838  &  43  &  19.5 \\
LA  &  872  &  41  &  21.3 \\ 
ME  &  567  &  16  &  35.4 \\
MD  &  2372  &  238  &  10.0 \\
MA  &  2557  &  336  &  7.6 \\
MI  &  1070  &  67  &  16.0 \\
MN  &  1208  &  26  &  46.5 \\
MS  &  507  &  24  &  21.1 \\
MO  &  997  &  34  &  29.3 \\ 
MT  &  803  &  2  &  401.5 \\
NE  &  1363  &  9  &  151.5 \\
NV  &  3194  &  10  &  319.4 \\
NH  &  779  &  57  &  13.7 \\
NJ  &  3102  &  470  &  6.6 \\
NM  &  1310  &  6  &  218.4 \\
NY  &  9371  &  162  &  57.8 \\
NC  &  790  &  79  &  10.0 \\
ND  &  1139  &  4  &  284.6 \\
\end{tabular}
\end{center}

\newpage

\begin{center}
\begin{tabular}{c|r|r|r}
State &	Population-weighted $\rho_W$ & Standard $\rho_S$ & Variability $V$ \\
\hline
OH  &  1134  &  109  &  10.4 \\
OK  &  941  &  22  &  42.8 \\
OR  &  1724  &  16  &  107.7 \\
PA  &  1973  &  110  &  17.9 \\
RI  &  2083  &  394  &  5.3 \\
SC  &  690  &  62  &  11.1 \\
SD  &  928  &  4  &  232.0 \\
TN  &  757  &  61  &  12.4 \\
TX  &  1848  &  40  &  46.2 \\
UT  &  2004  &  14  &  143.1 \\
VT  &  614  &  26  &  23.6 \\
VA  &  1617  &  81  &  20.0 \\
WA  &  1713  &  41  &  41.8 \\
WV  &  526  &  29  &  18.1 \\
WI  &  1253  &  41  &  30.6 \\
WY  &  871  &  2  &  435.4 \\
\end{tabular}
\end{center}

\end{document}